\documentclass[11pt,a4paper]{article}

\usepackage{authblk}
\usepackage{eurosym}
\usepackage{tabto}
\usepackage{enumerate}
\usepackage{amsmath}
\usepackage{enumitem}
\usepackage{amsfonts}
\usepackage{amssymb}
\usepackage{graphicx}
\usepackage{color}
\usepackage[pdfpagelabels]{hyperref}

\definecolor{mycolorpink}{rgb}{1, 0, 1}

\def\C{\mathcal{C}}

\usepackage{amsmath,amssymb,amsthm}
\usepackage{mathtools}
\newtheorem{theorem}{Theorem}[section]

\newtheorem{proposition}[theorem]{Proposition}

\newtheorem{remark}[theorem]{Remark}

\usepackage{graphicx,caption,subcaption,float}
\usepackage{booktabs}
\usepackage{listings}
\usepackage{float}
\usepackage{subcaption}

\usepackage[backend=bibtex,style=alphabetic]{biblatex}
\DeclareMathOperator*{\argmin}{argmin}

\def\R{\mathbb{R}}
\def\E{\mathbb{E}}

\begin{document}

\numberwithin{equation}{section}

\title{\textbf{Pricing options on illiquid assets \\ using liquid market benchmarks: \\ an application to energy markets}}
\author[1]{F. Aluigi\thanks{\texttt{federico.aluigi@enel.com}}}
\author[2,4]{L. Caramellino\thanks{\texttt{caramell@mat.uniroma2.it}}}
\author[3,4]{P. Pigato\thanks{\texttt{paolo.pigato@uniroma2.it}}}
\author[1,2]{E. Scrima\thanks{\texttt{edoardo.scrima3@enel.com}}}

\affil[1]{\small{Enel Global Energy and Commodity Management}}
\affil[2]{\small{Dipartimento di Matematica, Universit\`a di Roma Tor Vergata}}
\affil[3]{\small{Dipartimento di Economia e Finanza, Universit\`a di Roma Tor Vergata}}
\affil[4]{\small{INdAM-GNAMPA, Research Unit at Tor Vergata}}
\date{}
\maketitle
\begin{abstract}
\noindent
The Gasoil options market is illiquid, making it difficult to construct its implied volatility surface directly. However, it is closely linked to the highly liquid Brent options market.
In this paper, we jointly model Brent and Gasoil futures prices through a correlated Bachelier local volatility model: the Brent factor is described by a normal mixture diffusion model, while the Gasoil-Brent spot volatility spread is estimated using a data-driven procedure that identifies clusters of historical crack-spread levels and Gasoil-Brent volatility spreads.
The resulting bivariate model allows us to compute an implied volatility correction that maps Brent implied volatilities to Gasoil implied volatilities without using illiquid Gasoil option prices as inputs.
Monte Carlo simulations demonstrate that the resulting implied volatilities closely match observed Gasoil implied volatilities when benchmarked against more direct  approaches. These results suggest that the proposed framework is well suited for modeling refined products and pricing the corresponding financial derivatives.

\end{abstract}

\section{Introduction}
\label{introduction}

Constructing implied volatility surfaces is a central problem in derivatives pricing and risk management. In liquid markets, it is possible to infer a volatility surface directly from options market quotes, possibly after applying suitable interpolation, smoothing and no-arbitrage constraints \cite{gatheral2006volatility, Guo2016, gatheral2014arbitrage}. In illiquid markets, however, option quotes may be sparse, irregular or affected by wide bid--ask spreads. In these cases, direct calibration becomes unstable and model uncertainty becomes more relevant \cite{cont2006model}.

This issue is particularly important in energy markets. Crude oil benchmarks such as Brent typically benefit from more liquid futures and option markets, while refined products such as Gasoil may exhibit less reliable option quotes across strikes and maturities. As a consequence, constructing a Gasoil implied volatility surface from Gasoil option prices alone can be difficult, and it can be beneficial to use information from a related and more liquid market, such as Brent.

A natural economic link between crude oil and refined products is provided by the crack spread, usually interpreted as the price difference between a refined product and crude oil. Crack spreads are widely used as indicators of refining margins and play an important role in hedging and risk management in energy markets \cite{eia2011crack, haigh2002crack}. The literature on energy derivatives has studied spread options and crack spread options, highlighting the importance of modelling the joint behaviour of crude and refined product prices \cite{carmona2003spread,paschke2009pricing, mahringer2015crack}.  More generally, commodity and energy futures markets display specific features such as term-structure effects, time-varying volatility, seasonality and market incompleteness \cite{pindyck2001dynamics,benth2008stochastic}.

In practice, when an asset is illiquid but strongly related to a liquid proxy, market participants often transfer information from the liquid market to the illiquid one. From a theoretical perspective, this is connected to pricing in incomplete markets, where utility-based and proxy-hedging methods provide a possible framework for dealing with non-traded or illiquid risks \cite{henderson2009utility,callegaro2017utility,halperin2014pricing}. Our approach follows this proxy-based idea, but does not rely on utility indifference pricing. In fact, the Brent option market provides the reference volatility structure, while Gasoil futures prices enter through the Gasoil-Brent crack spread, which is used to construct a Gasoil-specific volatility correction.
We note that the transformation from Brent to Gasoil implied volatilities is not simply a level shift: skew and curvature are also affected. For instance, practitioners typically expect the right wing of Gasoil implied volatility smiles to be further away from the Brent one, with respect to the at-the-money point. An accurate model should be able to reproduce these features of implied skew and curvature. For general background on implied volatility shapes, skew and curvature, see \cite{gatheral2006volatility}.
\medskip

{\bf Our contribution.} In this paper we develop a framework for constructing the Gasoil volatility surface by exploiting the liquid Brent option market. We model Brent and Gasoil futures prices, denoted by \(B(\cdot,T)\) and \(G(\cdot,T)\), with a two-dimensional correlated extended Bachelier model. The use of a normal rather than log-normal specification is natural in commodity markets, where price changes are often modelled in absolute rather than proportional terms, and where negative prices cannot be ruled out in extreme market conditions. This choice is also consistent with the use of normal-type models in oil derivatives markets \cite{bouchouev2023virtual}.

For the Brent component, we adapt the mixture approach of Brigo, Mercurio and Sartorelli \cite{brigo2003alternativeasset,brigo2006interest}. In our setting, the Brent local volatility function \(\sigma^B(t,x)\) is chosen so that the marginal distribution of the futures price is represented by a mixture of Gaussian densities. This provides enough flexibility to reproduce volatility smiles while retaining a tractable structure of the dynamics.
See \cite{alexander2004normal} for applications of mixture diffusion models to volatility smiles.

The Gasoil volatility is then specified as a Brent-based volatility plus a Gasoil-specific correction:
\[
\sigma^G(t,T)
=
\sigma^B(t,B(t,T))
+
h\bigl(t,G(t,T)-B(t,T)\bigr).
\]
The function \(h\) depends on time and on the Gasoil--Brent crack spread. It is constructed through a data-driven clustering procedure that summarizes representative historical configurations of crack-spread levels and Gasoil--Brent volatility spreads.

Since the resulting Gasoil dynamics do not lead to closed-form option pricing formulas, we compute Gasoil option prices  by Monte Carlo simulation of the joint Brent--Gasoil system. 
We then obtain the corresponding Gasoil implied volatilities  by inverting the Bachelier pricing formula.

Finally, we link the reconstructed Gasoil surface to the observed Brent market surface. Denoting by \(\sigma^{B,\mathrm{mkt}}_{{\mathrm{IV}}}\) the Brent market-implied volatility surface and by \(\sigma^{G,\mathrm{mod}}_{{\mathrm{IV}}}\), \(\sigma^{B,\mathrm{mod}}_{{\mathrm{IV}}}\) the model-implied Gasoil and Brent volatility surfaces, we define
\[
\sigma^{G,\mathrm{final}}_{{\mathrm{IV}}}(K,T)
=
\sigma^{B,\mathrm{mkt}}_{{\mathrm{IV}}}(K,T)
+
\Bigl(
\sigma^{G,\mathrm{mod}}_{{\mathrm{IV}}}(K,T)
-
\sigma^{B,\mathrm{mod}}_{{\mathrm{IV}}}(K,T)
\Bigr).
\]
This construction preserves the level of the liquid Brent market surface and adds the Gasoil--Brent implied volatility spread generated by the model.

\smallskip

This paper makes therefore the following main contributions. 

\smallskip

First, it adapts the log-normal mixture framework of Brigo and Mercurio to Energy markets.

\smallskip

Secondly, it proposes a framework for constructing volatility surfaces in illiquid markets by consistently transferring volatility information from a liquid benchmark (Brent) to a related illiquid asset (Gasoil). Furthermore, it introduces a data-driven volatility-spread model based on the crack spread, with the function $h$ constructed from representative historical market configurations. Numerical simulations show that the framework reproduces key features of the Gasoil volatility surface. 

\smallskip

Finally, the model relies on Brent option quotes, Brent and Gasoil futures prices, and the historical crack-spread/volatility-spread configurations used to construct \(h\). This makes the framework suitable for situations in which Gasoil option data are less liquid or less reliable than Brent option data. The resulting model is suited for pricing and risk management applications, as well as for analysing relative value opportunities between crude oil and refined products.

We conclude observing that, at the validation dates, the Gasoil option surface is not used as an input to construct the model-implied surface. It is used only ex post as a benchmark to analyse the quality of the reconstruction.

The paper is organised as follows. Section~\ref{The model} first introduces the model and describes the data, the construction of the rolling futures series, and the function $h$. Section~\ref{Calibration} presents the calibration of the Brent volatility surface and the correlation parameter. Section~\ref{Gasoil option pricing} discusses the joint dynamics and the Monte Carlo pricing procedure. Section~\ref{sec:numerical_results}
 reports the numerical results.

\smallskip

\textbf{Acknowledgments.}
L.C. and E.S. acknowledge support from the MUR Excellence Project MatMod@TOV awarded to the Department of Mathematics, Tor Vergata University of Rome, CUP E83C23000330006; L.C and P.P. acknowledge support from the  the Research Project MLFGTSRL from Tor Vergata University of Rome, CUP E83C25000470005.

\section{The model}
\label{The model}

Our goal is to build a model for pricing and hedging options on Gasoil futures, whose market is closely related to Brent futures market, but relatively illiquid compared to it. We do so using options on Brent futures, which are widely traded on regulated exchanges and backed by full volatility surfaces, as a proxy.

We work with two different market instruments: 
\begin{itemize}
    \item \textbf{Futures contract.} 
    A futures contract is an agreement between two parties to buy or sell an asset at a future date at a price fixed today. In our case, the underlying asset is a commodity. We focus on fixed-delivery futures contracts, meaning that delivery takes place at a future date $T$, specified at the issuance of the contract.

    \item \textbf{Option on a futures contract.} 
    An option on a futures contract is a contract that gives its holder the right, but not the obligation, to buy (call option) or sell (put option) the underlying futures contract at a predetermined price $K$, called the strike price. The date at which this right can be exercised is called the option maturity, or expiration date, and is denoted by $\widetilde T$.
We will assume that the date is specified at issuance (European option).

    For the contracts considered in this work, the option maturity is close to the futures delivery date: the option expires approximately two working days before the futures maturity,
    $$
    \widetilde T \simeq T.
    $$
    For this reason, and to simplify the notation, we do not distinguish between the option maturity and the futures maturity in the rest of the analysis.
\end{itemize}

The model we propose for the dynamics of Brent futures with maturity $T$ is as follows:
at time $t \in [0, T],$ the Brent futures $B(t,T)$ satisfies
\begin{equation}
\label{Brent_model_proposed}
\left\{
	\begin{array}{lll}
		dB(t,T) = \sigma^B(t,B(t,T)) dW^{B}_t \\
        B(0,T) \text{ given},
	\end{array}
\right.
\end{equation}
where $W^{B}$ is a standard Brownian motion and the \emph{local volatility} $\sigma^B$ is a well-behaved deterministic function.
Then, we assume a dynamics for Gasoil futures with maturity $T$ given by 
\begin{equation}
\label{Gasoil_model_proposed}
\left\{
	\begin{array}{ll}
		dG(t,T) = \sigma^G(t,T) dW^{G}_t \\
        G(0,T) \text{ given},
	\end{array}
\right.
\end{equation}
where $W^{G}$ is a standard Brownian motion correlated with $W^{B}$,
$$
d\langle W^{G}, W^{B} \rangle_t = \rho^{G,B}_T dt,
$$
with $|\rho^{G,B}_T| \le 1.$  Brent volatility provides a natural benchmark for the overall market uncertainty, to which we add a correction that represents additional variability specific to the Gasoil market, reflecting e.g. refining margins. We assume that
\begin{equation}
\label{sigmaG}
\sigma^G(t,T) = \sigma^B(t,B(t,T)) + h\bigl(t, G(t,T) - B(t,T)\bigr),
\end{equation}
where $h$ is a deterministic function depending on time and on the Gasoil-Brent \emph{crack spread} $G(t,T) - B(t,T)$.  In our specification $h$ depends solely on time and the crack spread and does not introduce additional sources of randomness. It will be determined using a data-driven (machine learning) approach (see Section \ref{h_function_section}).

\subsection{The dynamics of Brent futures}
\label{the_spread_model}

We specify here the Brent futures model \eqref{Brent_model_proposed}. We do so by adapting the log-normal mixture framework of Brigo and Mercurio \cite{brigo2003alternativeasset,brigo2006interest} and reformulating it  in terms of normal processes: we choose the volatility function $\sigma^B$ in \eqref{Brent_model_proposed} so that the probability density of $B(t,T)$ becomes a mixture of Gaussian densities (cf. also \cite{alexander2004normal}). 

For $p\geq 1$ we denote by $L^p([0,T])$ the set of (deterministic) functions $f:[0,T]\to\R$ such that $\int_0^T|f(t)|^pdt<\infty$.

\begin{proposition}
\label{mixture_proposition}
Let $N\in\mathbb{N}$, with $N\geq 2$. Let:
\begin{enumerate}
    \item $\lambda_1,\ldots,\lambda_N\in (0,1)$ such that $\lambda_1+\cdots+\lambda_N=1$;
    \item $\mu_1,\ldots,\mu_N\in L^1([0,T])$;
\item $\sigma_1,\ldots,\sigma_N\in L^2([0,T])$.
\end{enumerate}
Let $p^i_t$, $i=1,\ldots, N$ and $t\in[0,T]$, denote the density function of the Gaussian law with mean and variance given by
\begin{equation}\label{M-V}
M^{i,T}_t=B(0,T)+\int_0^t \mu_i(s)ds \quad\mbox{and}\quad
V^i_t=\int_0^t\sigma^2_i(s)ds
\end{equation}
respectively, where $B(0,T)$ denotes the initial condition of \eqref{Brent_model_proposed}. Then, the Gaussian mixture density
\begin{equation}
    \label{B_density_condition}
    p_t(y) := \sum_{i=1}^{N} \lambda_i p_t^i(y), \quad t\in[0,T],
\end{equation}
is the transition density of the solution to \eqref{Brent_model_proposed} when the squared volatility is given by
\begin{equation}\label{Brent-local-square-vol}
    \sigma^B(t,y)^2 = \frac{\sum_{i=1}^{N}\lambda_i p_t^i(y)(\sigma_i(t))^2}{\sum_{i=1}^{N}\lambda_i p_t^i(y)} + \frac{\sum_{i=1}^{N}2\lambda_i \mu_i(t) \int_{y}^{+\infty} p_t^i(x)dx}{\sum_{i=1}^{N}\lambda_i p_t^i(y)}.
\end{equation}

\end{proposition}

\begin{proof}
We just need to prove that the density $p_t$ in \eqref{B_density_condition} solves the Fokker-Planck equation associated to equation \eqref{Brent_model_proposed} when the volatility $\sigma^B$ is given by formula \eqref{Brent-local-square-vol}, that is
$$
    \frac{\partial}{\partial t}p_t(y)  = \frac{1}{2}\frac{\partial^2}{\partial y^2} \big(p_t(y)(\sigma^B(t, y))^2 \big) .
$$
For $i=1,\ldots,N$, consider the diffusion process
\begin{equation}
    \label{N_diff_processes_dynamics}
    \left\{
	\begin{array}{ll}
		dS_i(t) = \mu_i(t) dt + \sigma_i(t) dW_t \\
        dS_i(0) = B(0,T).
	\end{array}
\right.
\end{equation}
In particular, the marginal density of $S_i(t)$ is given by $p^i_t$, so it satisfies the associated Fokker-Planck equation:
$$
    \frac{\partial}{\partial t}p_t^i(y)  = -  \frac{\partial}{\partial y}\big(\mu_i(t) p_t^i(y)\big) + \frac{1}{2}\frac{\partial^2}{\partial y^2} \big(p_t^i(y)(\sigma_i(t))^2 \big).
$$
Therefore,
\begin{align*}
\frac{\partial}{\partial t}p_t(y) 
& =\sum_{i=1}^N{\lambda_i}\frac{\partial}{\partial t}p^i_t(y)\\
&=\sum_{i=1}^N{\lambda_i} \Big(-\frac{\partial}{\partial y}\big(\mu_i(t) p_t^i(y)\big)
+ \frac{1}{2}\frac{\partial^2}{\partial y^2} \big(p_t^i(y)(\sigma_i(t))^2 \big)\Big)\\
&=-\frac{\partial}{\partial y}\Big(\sum_{i=1}^N{\lambda_i} \mu_i(t) p_t^i(y)\Big)
+ \frac{1}{2}\frac{\partial^2}{\partial y^2} \Big(\sum_{i=1}^N{\lambda_i} p_t^i(y)(\sigma_i(t))^2 \Big)
\end{align*}
Now, taking $p_t$ and $(\sigma^B)^2$ as in formula \eqref{B_density_condition} and \eqref{Brent-local-square-vol} respectively, straightforward computations give 
$$
-\frac{\partial}{\partial y}\Big(\sum_{i=1}^N{\lambda_i} \mu_i(t) p_t^i(y)\Big)
+ \frac{1}{2}\frac{\partial^2}{\partial y^2} \Big(\sum_{i=1}^N{\lambda_i} p_t^i(y)(\sigma_i(t))^2 \Big)
=\frac{1}{2}\frac{\partial^2}{\partial y^2} \big(p_t(y)(\sigma^B(t, y))^2 \big)
$$
and the statement holds.

\end{proof}

\begin{remark}\label{rem-lambda-mu}
\rm
We note that the parameters $\lambda_i$ and the functions $\mu_i$, $i=1,\ldots,N$, must satisfy 
the condition
\begin{equation}
     \label{mu_constraint}
     \sum_{i=1}^{N} \lambda_i \mu_i(t) = 0, \ \ \text{for each $t\in [0,T]$}.
\end{equation}
Using \eqref{Brent_model_proposed} and the fact that $p_t$ in \eqref{B_density_condition} is the probability density function of $B(t,T)$,
\begin{align*}
\E(B(t,T))=B(0,T) =\int_{-\infty}^{+\infty} yp_t(y)dy
=\sum_{i=1}^N\lambda_i \int_{-\infty}^{+\infty} yp^i_t(y)dy.
\end{align*}
Since $p^i_t$ denotes the density function of the Gaussian law with mean $B(0,T)+\int_0^t \mu_i(s)ds$, 
\begin{align*}
B(0,T) = B(0,T)+\sum_{i=1}^N\lambda_i \int_0^t \mu_i(s)ds
\end{align*}
and \eqref{mu_constraint} follows.

\end{remark}

While the functions $\mu_i$ cannot be arbitrarily chosen because they have to respect the constraint in \eqref{mu_constraint}, 
the volatility functions $\sigma_i(\cdot), \ \forall i=1, \dots, N$ are free, that is, they can be arbitrarily set. Hence, we consider the case where $\sigma_i$ is piecewise constant.

\def\ss{\gamma}
Let us suppose to work with $M$ different tenors (i.e. maturities of the futures contracts) $0=T_0 < T_1 < T_2 < \dots < T_M$. 
For $i=1,\ldots,N$, we consider parameters $\ss_{i,j}>0$, $j=1,\ldots,M$.
 Then, for the maturity $T=T_m$, $m\in\{1,\ldots,M\}$, we set the volatility function on the time horizon $[0,T_m]$ as follows
\begin{equation}\label{sigma-i}
    \sigma_{i,m}(t) = \sum_{i=1}^{m} \ss_{i,j} \mathbb{I}(t)_{(T_{j-1}, T_{j}]},\quad t\in[0,T_m],
\end{equation}
in which $T_0=0$. Note that $\sigma_{i,m}$ is the restriction of $\sigma_{i,M}$ over $[0,T_m]$. See
Figure \ref{fig:sigma-i} for an example.

\begin{figure}[H]
    \centering
    \includegraphics[width=0.5\linewidth]{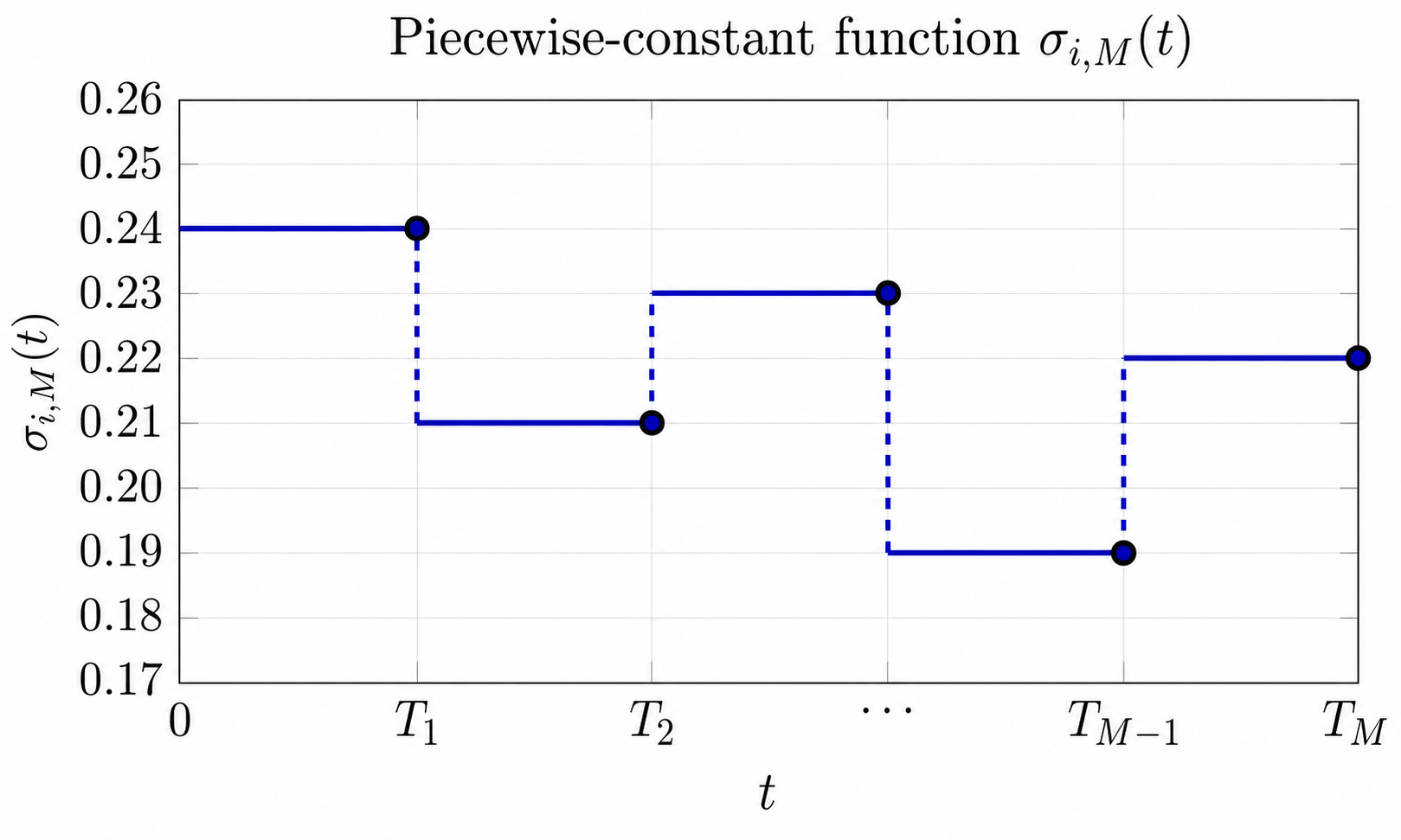}
    \caption{An example of a piecewise-constant volatility function $\sigma_i$}
    \label{fig:sigma-i}
\end{figure}
\noindent
Note that $\sigma_{i,m}$ defines a maturity-dependent volatility function, that captures the main term-structure effects observed in the market, reflecting different volatility conditions across short, medium, and long horizons. The volatility is constant within each interval $(T_{k-1},T_k]$ (time bucket) and may change (upward or downward) only at the bucket boundaries.

\subsection{The dynamics of Gasoil futures}

From \eqref{Gasoil_model_proposed} and \eqref{sigmaG}, for the Gasoil futures model it remains to specify the function $h$. 

\subsubsection{Rolling futures time series}
\label{sec:rolling_series}

In futures markets, commodities are traded through multiple contracts with different maturities rather than through a single continuously traded price. Since each futures contract expires and may lose liquidity near delivery, using a single contract does not provide a consistent basis for analysing historical returns or volatility over time. Therefore, historical analysis requires the construction of a continuous price series that remains representative of the liquid market.

A \textit{rolling futures series} is built exactly for this purpose. The idea is to follow one contract for a period, then to switch to a contract with later delivery, and repeat this over time. Formally, 
given $F(t,T)$  the price at time $t$ of the futures contract that delivers at $T$,
we define a rule $R(t)$ that chooses, at each date $t$, which maturity (which contract) we use. The rolling series is then
$$
\widetilde{F}(t) := F\bigl(t, R(t)\bigr).
$$
The rule $R(t)$ is typically piecewise constant over time and changes only at a set of roll dates $\tau_1, \tau_2, \ldots$, usually corresponding to (or closely approximating) the maturity dates.  At each roll date we move from the current contract to the next one (or to the most liquid one), so that the series remains continuous in time and representative of the market.

This rolling idea is also useful when we work with a whole tenor grid $0=T_0<T_1<\cdots<T_M$. In that case we want to track several points of the futures curve, not just one contract. We can therefore build one rolling series for each tenor. For every $T_m$ we consider, depending on the problem, a tenor-specific rule $R_m(t)$ and the corresponding rolling series
$$
\widetilde{F}_m(t) := F\bigl(t, R_m(t)\bigr), \qquad m=1,\ldots,M,\quad t\geq 0.
$$
So, for each commodity we obtain $M$ rolling series, one for each tenor. Let us illustrate this with a stylized example.

Consider futures contracts with expiry $m$-months ahead and, for simplicity, assume that each months is 30 days. This gives maturities $T_k=k\times 30$, $k\geq 1$. Let now $\C_m$ denote the contract with maturity $T_m$. We aim to construct a historical synthetic rolling day-by-day series $\tilde F_m(t)$ for this contract.
A way to solve this problem is the following:
\begin{itemize}
	\item[-] for $t=1,\ldots, 30$, we set $\widetilde{ F}_m(t)=F(t, T_m)$;
	\item[-] for $t=30+1,\ldots, 30+30$, we set $\widetilde{ F}_m(t)=F(t, T_m+30)=F(t, T_{m+1})$;
	\item[] $\vdots$
	\item[-] for $t=k\times 30+1,\ldots, k\times 30+30$, we set $\widetilde{ F}_m(t)=F(t, T_m+k\times 30)=F(t, T_{m+k})$.
\end{itemize}

The data $F(t, T_{m+k})$ are taken from past contracts, possibly already expired, with maturity  $T_{m+k}$. 

In practice, months do not last exactly 30 days and contracts typically become less liquid
when approaching delivery. Therefore, the roll is performed a few business days before
maturity, for example two or three days before expiry. In addition, the raw series obtained
by simply switching from one contract to the next may contain artificial jumps at roll
dates, because the two contracts usually have different prices. Following the
single-futures-roll construction of \cite{lopez2018advances}, these jumps can be removed
by computing cumulative roll gaps and subtracting them from the raw spliced price
series. These are the rolling series we use to compute past tenor-by-tenor data and to
build term-structure features for the rest of the methodology.

\subsubsection{The crack spread function $h$}
\label{h_function_section}

We see from equation \eqref{sigmaG} that Gasoil dynamics incorporate 
both market-implied information, through the Brent volatility function $\sigma^B$, and historical information, through the function $h$.

Let $0=T_0<T_1<\cdots<T_M$ be the tenor grid under consideration. Throughout this section, all historical quantities are computed on the tenor-specific rolling futures series as described in Section~\ref{sec:rolling_series}. In particular, for each tenor $T_m$ we work with the rolling (roll-adjusted) series $B_m(t)$ and $G_m(t)$, which provide a consistent historical record for the analysis of past returns and historical volatilities of Brent and Gasoil.

Let $\mathcal{T}=\{t_0,t_1,\dots,t_L\}$ denote the set of observation dates, with $t_0<t_1<\cdots<t_L$, which we assume to be consecutive trading days (otherwise, one has to normalize with respect to the time variation).

\paragraph{Step 1. Returns and crack spreads.}

Fix a tenor $T_m$. For each observation date $t_k\in\mathcal{T}$, define the one-period arithmetic returns
$$
r_B(t_k,T_m) := B_m(t_k)-B_m(t_{k-1}),
\qquad
r_G(t_k,T_m) := G_m(t_k)-G_m(t_{k-1}),
$$
for all $k=1,\dots,L$.
Next, define the crack spread level at date $t_k$ and tenor $T_m$ as
$$
S(t_k,T_m) := G_m(t_k)-B_m(t_k),
\quad m=1,\dots,M.
$$
This definition is consistent with the cluster-based construction below, since both the crack levels and the historical descriptors are computed from rolling time series constructed using the same rolling procedure for Brent and Gasoil, including identical contract-switching rules.

\paragraph{Step 2. Rolling statistical descriptors.}
\label{Step 2. Rolling statistical descriptors}
Fix a window length $n\in\mathbb{N}$. For each date $t_k\in\mathcal{T}$ such that the rolling window is well defined, and for each tenor $T_m$, define the historical volatilities
\begin{align}
&\hat{\sigma}_B(t_k,T_m)
:=
\operatorname{std}\big(
r_B(t_u,T_m):u=k-n+1,\dots,k
\big),\label{sigmahat-B}\\
&\hat{\sigma}_G(t_k,T_m)
:=
\operatorname{std}\big(
r_G(t_u,T_m):u=k-n+1,\dots,k
\big)\label{sigmahat-G}.
\end{align}

We also define the rolling average crack level
$$
\bar{S}(t_k,T_m)
:=
\frac{1}{n}\sum_{u=k-n+1}^{k} S(t_u,T_m),
$$
and the relative historical volatility spread
$$
V(t_k,T_m) := \hat{\sigma}_G(t_k,T_m)-\hat{\sigma}_B(t_k,T_m).
$$
Hence, $V(t,T_m)>0$ indicates that Gasoil has recently been more volatile than Brent at tenor $T_m$ and viceversa.

\paragraph{Step 3. Feature-space representation.}

For each admissible date $t  \in \mathcal{T}$, define the feature vector
$$
x_t
:=
\Big(
V(t,T_1),\dots,V(t,T_{M}),
\bar{S}(t,T_1),\dots,\bar{S}(t,T_{M})
\Big)
\in \mathbb{R}^{2M}.
$$
Thus, each date is mapped into a \(2M\)-dimensional feature space collecting: 
(i) the term structure of relative historical volatility spreads, and 
(ii) the term structure of average crack levels.
Since these quantities are computed over rolling windows, consecutive observations are based on largely overlapping data and are therefore strongly correlated. To reduce this mechanical dependence we downsample the sequence \(\{x_t\}_{t\in\mathcal T}\) by retaining one observation every \(n_2\) dates.

In our implementation, we set \(n_2 \approx n/2\), where \(n\) is the length of the rolling window. This choice reduces the overlap between consecutive retained observations while preserving a sufficiently large sample for clustering.

The clustering is therefore performed on the reduced set
\[
\{x_t\}_{t\in\mathcal T'},
\qquad
\mathcal T' \subset \mathcal T.
\]
To avoid unnecessary notation, however, we continue to denote the feature sample by
\[
\{x_t\}_{t\in\mathcal T}.
\]

\paragraph{Step 4. Clustering of historical crack-spread and volatility-spread configurations: cityblock $k$-medians specification.}

Let
$$
X=\{x_1,\dots,x_L\}\subset \mathbb{R}^{2M}
$$
denote the set of feature vectors obtained from the admissible, possibly downsampled, observation dates. We define representative market configurations through a clustering procedure with a fixed number $k$ of clusters; in our implementation (see Section \ref{Calibration}), we take $k=6$. 
	We empirically observe that this number is large enough to capture different market configurations and not so large as to produce noticeable overfitting.

We now formalize the clustering problem. For each observation $x_i\in X$ and each cluster $j\in\{1,\dots,k\}$, introduce the binary allocation variables
$$
z_{ij}\in\{0,1\},
\qquad i=1,\dots,L,\quad j=1,\dots,k,
$$
where $z_{ij}=1$ if and only if $x_i$ is assigned to cluster $j$, and $z_{ij}=0$ otherwise. Each observation $x_i$ is assigned to exactly one cluster, so that
$$
\sum_{j=1}^k z_{ij}=1,
\qquad i=1,\dots,L.
$$

For each cluster $j$, let $c_j\in\mathbb{R}^{2M}$ denote its representative vector. The clustering criterion is based on the cityblock distance
$$
d_1(\xi,\eta):=\sum_{\ell=1}^{2M} |\xi_\ell-\eta_\ell|,
\qquad \xi,\eta\in\mathbb{R}^{2M}.
$$
The clusters $\{\bar z_{i,j}\}_{i,j}$ and their representative vectors $\bar c_1,\ldots,\bar c_k$ are then defined as a solution of the optimization problem
$$
(\{\bar z_{i,j}\}_{i,j}, \bar c_1,\ldots,\bar c_k)\in \argmin_{\substack{\{z_{ij}\}_{i,j},\, c_1,\dots,c_k :\\
z_{ij}\in\{0,1\}, \sum_{j=1}^k z_{ij}=1	}}
\sum_{i=1}^L \sum_{j=1}^k z_{ij}\, d_1(x_i,c_j).
$$
It can be shown \cite{hastie2009elements} that, for $j=1,\ldots,k$,  the optimal representative vector $\bar c_j$ is obtained coordinate-wise as the median of the observations assigned to cluster $j$.
 Hence, the above specification corresponds to a $k$-medians procedure, rather than to the classical Euclidean $k$-means, which makes the construction less sensitive to extreme observations.
We write the corresponding clusters as
$$
\mathcal C_j:=\{x_i\in X:\ \bar z_{ij}=1\},
\qquad j=1,\dots,k.
$$

\paragraph{Step 5. Structure of the centroids.}
Let $\bar c_j\in\mathbb{R}^{2M}$, $j=1,\dots,k$, denote the representative vectors obtained from the clustering problem above. Each $\bar c_j$ is decomposed as
\begin{equation}
\label{centroid}
\bar c_j=
\bigl(
c_j^{V}(T_1),\dots,c_j^{V}(T_M),c_j^{S}(T_1),\dots,c_j^{S}(T_M)
\bigr),
\end{equation}
where $c_j^{V}(T_m)$ and $c_j^{S}(T_m)$ denote, respectively, the representative relative historical volatility spread and the representative crack level at tenor $T_m$ associated with cluster $j$.

Thus, the first $M$ coordinates of $\bar c_j$ describe the typical term structure of the relative volatility spread within cluster $j$, while the last $M$ coordinates describe the corresponding term structure of the crack level.

\paragraph{The function $h$.}
Once the clustering procedure has been performed on the historical sample, the estimated centroids are regarded as fixed. For each cluster $j=1,\dots,k$, let
$$
\bar c_j=
\bigl(
c_j^{V}(T_1),\dots,c_j^{V}(T_M),c_j^{S}(T_1),\dots,c_j^{S}(T_M)
\bigr)\in\mathbb R^{2M}
$$
denote the corresponding representative vector, where $c_j^V(T_m)$ is the volatility-spread component and $c_j^S(T_m)$ is the crack-level component at tenor $T_m$.

At each time $t$ and for each tenor $T_m$, the correction function $h$ is defined by interpolating the volatility-spread components of the centroids as a function of their crack-level components. More precisely,
$$
h\bigl(t,S(t,T_m)\bigr)
=
\operatorname{Interp}\Bigl(
\{c_j^S(T_m)\}_{j=1}^k,\,
\{c_j^V(T_m)\}_{j=1}^k,\,
S(t,T_m)
\Bigr),
\qquad m=1,\dots,M,
$$
where $\operatorname{Interp}$ denotes the { linear interpolation} rule adopted in the numerical implementation.

Hence, $h(t,S(t,T_m))$ maps the current crack level $S(t,T_m)$ into a volatility correction obtained from the historical centroids. In this way, the centroids do not induce a purely discrete regime-switching rule; rather, they provide the interpolation nodes of a crack-dependent correction function.

\paragraph{Summary of the construction.}

The above procedure associates to each admissible observation date $t$ a state vector $x_t\in\mathbb{R}^{2M}$ that summarizes (i) the term structure of relative historical volatility and (ii) the term structure of crack levels. The set of such vectors is then partitioned into a finite number of clusters, each cluster being interpreted as a representative historical configuration of the Brent--Gasoil crack market. The final output therefore consists of a collection of $k$ centroids describing typical configurations of the pair
$$
\bigl(V(t,T_m),\bar{S}(t,T_m)\bigr)_{m=1}^{M},
$$
which we use as reduced-form descriptors of the market state and as inputs to build the Gasoil volatility function in \eqref{sigmaG}. 

So, for a fixed tenor $T_m$, each historical observation is represented in the two-dimensional plane by the pair
$$
\bigl(\bar S(t,T_m),\,V(t,T_m)\bigr),
$$
where $\bar S(t,T_m)$ denotes the rolling average Gasoil-Brent crack spread and $V(t,T_m)$ denotes the corresponding historical volatility spread. Clustering these observations groups similar historical crack-spread and volatility-spread configurations at the tenor level, see Figure \ref{fig:cluster_tenor_t3}.

\begin{figure}[H]
    \centering
    \includegraphics[width=0.75\textwidth]{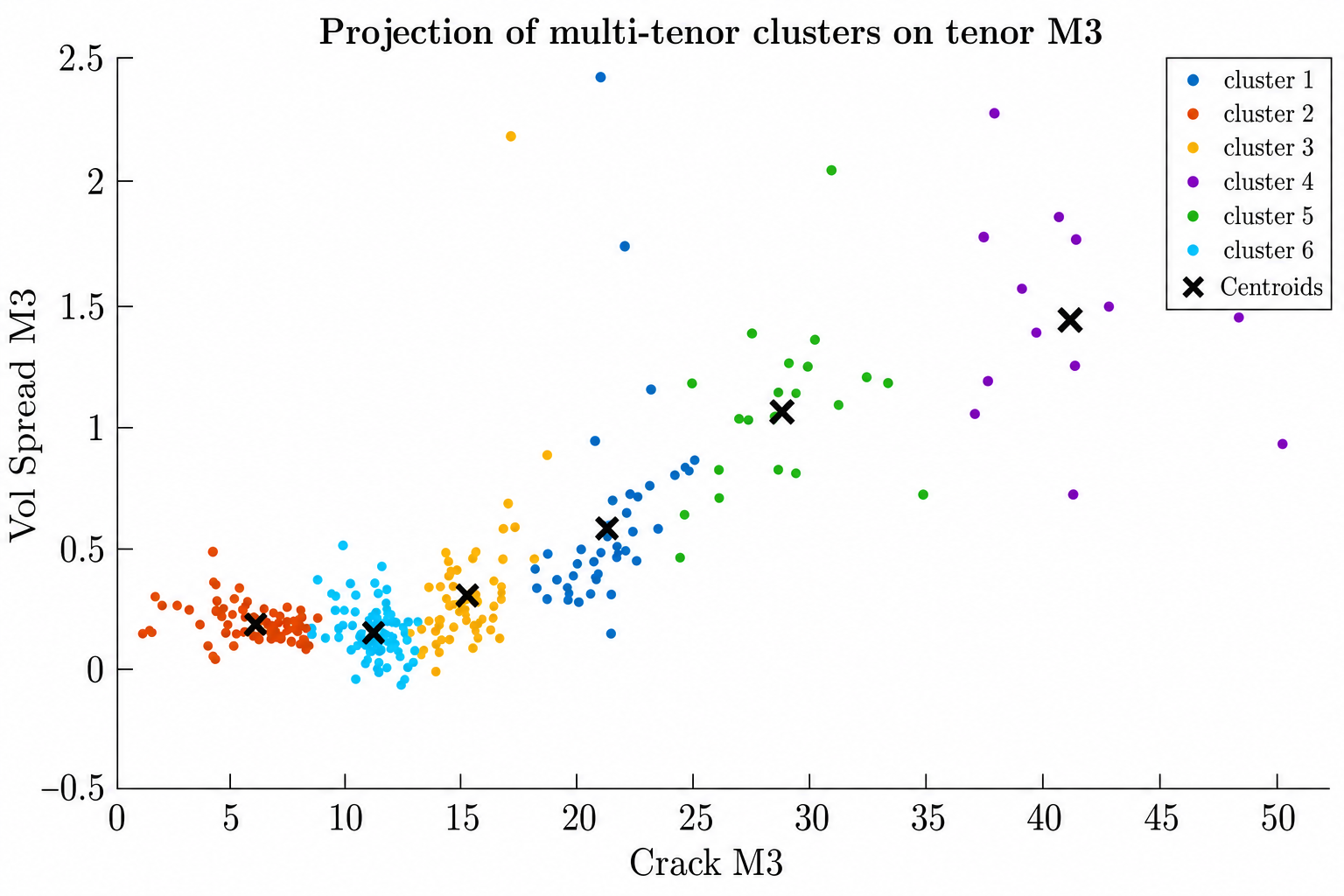}
    \caption{Scatter plot of the rolling average crack spread $\bar S(t,T_3)$ versus the historical volatility spread $V(t,T_3)$. Each point corresponds to one historical observation date. Colors identify the clusters obtained from the $k$-means procedure, while black crosses denote the corresponding cluster centroids.}
    \label{fig:cluster_tenor_t3}
\end{figure}

\section{Calibration}
\label{Calibration}
We now turn to the calibration of the processes, with the aim of recovering the parameters introduced in Section \ref{the_spread_model}: the Brent volatility $\sigma^B$, the correlation coefficients $\rho^{B,G}$ and  the spread-function $h$. We first estimate the Brent volatility function in \eqref{Brent_model_proposed} to observed European option prices, which are available since this is a liquid market.

\subsection{Brent dynamics calibration}
For the calibration of the Brent dynamics, 
we consider 
a normal mixture approach based on  Proposition \ref{mixture_proposition}.  We recall that European call/put options on Brent futures have, up to negligible time interval, the same maturity date.

Let \( X=\{X_t\}_{t\ge 0} \) denote a price process. For a fixed maturity \(T\), suppose that the marginal density of \(X_T\) under the risk neutral measure
can be written as
$$
p_T(y)=\sum_{i=1}^N \lambda_i\, p_T^{(i)}(y),
\qquad \lambda_i \ge 0,\quad \sum_{i=1}^N \lambda_i =1,
$$
where \(p_T^{(i)}\) are suitably chosen base densities for which European option pricing formulas are available in closed or semi-closed form.
Under this assumption, European option prices admit an immediate decomposition. Given a European payoff \(\Phi(X_T)\) and assuming that the interest rate $r=0$, one has 
\[
O(0;K,T)
=
\mathbb{E}
\!\left[\Phi(X_T)\right]
=
\int_\R \Phi(y)\,p_T(y)\,dy
=
\sum_{i=1}^N \lambda_i\, O^{(i)}(0;K,T),
\]
where
\[
O^{(i)}(0;K,T)
:=
\int_\R \Phi(y)\,p_T^{(i)}(y)\,dy
\]
is the option price generated by the \(i\)-th base model. Hence, the price of a European option is expressed as a convex combination of the prices of the same contract evaluated under simpler auxiliary dynamics.  In our setting, they are given by Bachelier specifications with constant mean and variance parameters, which allows for an efficient calibration to market option data while accurately reproducing the observed implied volatility.

The market provides, for a collection of maturities $\{T_m\}_{m=1}^M$ and strikes $\{K_{m,j}\}_{j=1}^{J_m}$, quoted information in the form of implied volatilities $\sigma^{\mathrm{mkt}}_{{\mathrm{IV}}}(K_{m,j},T_m)$ associated with European call/put options on Brent futures.

	Let us note that in commodity
markets, option smiles are typically written on futures with different maturities, hence on
different underlying contracts, rather than on a single underlying asset. As noted by
\cite{Albani2018}, this makes it difficult to observe a full option-price surface for a
fixed futures. For this reason, and consistently with the concerns discussed by
\cite{bouchouev2023virtual} for oil options, we do not attempt to calibrate a fully
unconstrained two-dimensional volatility surface. Instead, we calibrate maturity-wise
marginal mixture distributions and impose only the consistency conditions required by the reconstruction of the dynamics.

To compute the model implied volatility, we exploit Proposition \ref{mixture_proposition}: for each quoted maturity $T_m$, we represent the marginal law of the Brent futures $B_{T_m}=B(T_m,T_m)$ as a Gaussian mixture: 
\begin{equation}\label{p-Tm}
p_{T_m}(y)=\sum_{i=1}^N \lambda_i\, p_{T_m}^{(i)}(y),
\qquad \lambda_i \ge 0,\quad \sum_{i=1}^N \lambda_i =1,\quad p_{T_m}^{(i)}\sim \mathcal N\!\left(M^{i,m}_{T_m},\,V^{i,m}_{T_m}\right),
\end{equation}
where the notation $p\sim \mathcal N\!\left(M,\,V\right)$ means that $p$ is the Gaussian probability  density with mean $M$ and variance $V$. In \eqref{p-Tm}, $M^{i,m}_{T_m}$ and $V^{i,m}_{T_m}$ are given in \eqref{M-V} with the choice $T=T_m$.
In particular, the dependence on the maturity index $m$ is made explicit, so that both the mean component and the cumulative variances are allowed to vary across quoted tenors. Moreover, from Remark \ref{rem-lambda-mu}, we impose
\[
\sum_{i=1}^N \lambda_i\,M^{i,m}_{T_m}=B(0,T_m).
\qquad m=1,\dots,M,
\]

For each component $i$ and maturity $T_m$, we write
\[
V^{i,m}_{T_m}=\bar{\sigma} _{i,m}^2\,T_m,
\]
so that $V^{i,m}_{T_m}$ is interpreted as the cumulative variance of the $i$-th component up to $T_m$.

We denote by $\Theta$ the full set of model parameters, given by $\lambda_i, \bar \sigma^{i,m}$ and $\mu_i$,  $i=1,\ldots,N$, satisfying all constraints just described.  
Denoting by $\sigma^{\mathrm{model}}_{{\mathrm{IV}}}(K,T;\Theta)$ the implied volatility obtained by inverting the model price $O^{\mathrm{model}}(K,T;\Theta)$, the calibration is then carried out by minimizing the least-squares objective in implied-volatility space
\[
\min_{\Theta}\sum_{m=1}^M\sum_{j=1}^{J_m}
\left(
\sigma^{\mathrm{model}}_{{\mathrm{IV}}}(K_{m,j},T_m;\Theta)-\sigma^{\mathrm{mkt}}_{{\mathrm{IV}}}(K_{m,j},T_m)
\right)^2.
\]
After calibration, we can recover the volatility functions $\sigma_{i,m}$ in \eqref{sigma-i}.
The calibrated parameters $\bar{\sigma}_{i,m}$ give 
the cumulative variances $V^{i,m}_{T_m}$ which are in turn converted into piecewise-constant volatility levels over the successive maturity intervals by matching increments of total variance. Setting $T_0:=0$, we define
\[
\gamma_{i,1}^2=\frac{V^{i,1}_{T_1}}{T_1},
\qquad
\gamma_{i,m}^2=\frac{V^{i,m}_{T_m}-V^{i,m-1}_{T_{m-1}}}{T_m-T_{m-1}},
\qquad m=2,\dots,M.
\]
Equivalently,
\[
V^{i,m}_{T_m}=\sum_{k=1}^m \gamma_{i,k}^2\,(T_k-T_{k-1}),
\qquad m=1,\dots,M,
\]
so that the piecewise-constant volatility levels reproduce the calibrated total variance at each quoted maturity.

This piecewise-constant volatility representation is well defined whenever the calibrated cumulative variances are non-decreasing across maturities, 
\[
V^{i,m}_{T_m}\ge V^{i,m-1}_{T_{m-1}}, \qquad m=2,\dots,M.
\]
In practice, however, monotonicity of the cumulative variances is not imposed explicitly during the calibration step, so small local violations may occasionally arise.
{In such cases, the sequence $\{V^{i,m}_{T_m}\}_{m=1}^M$ should be interpreted
primarily as a discrete set of calibrated marginal variance parameters. The
piecewise-constant volatility representation is then an ex-post dynamic
reconstruction and is admissible only after enforcing non-negative variance
increments. This is done by applying a mild monotone projection to the calibrated
total variances before computing the levels $\gamma_{i,m}$. This step should not be
interpreted as recovering a unique global arbitrage-free volatility surface; it only
restores the internal consistency required to interpret $\gamma_{i,m}^2$ as variance
increments.}

\subsection{The correlation parameter}
	For each tenor $T_m$, the correlation $\rho^{G,B}_{T_m}$ between the Brownian motions driving the Brent and Gasoil futures dynamics must also be specified. This parameter cannot be calibrated from option prices, since Gasoil option prices are not quoted in the market. For this reason, we estimate it  from historical data, using rolling series of Brent and Gasoil futures. For each tenor $T_m$, we work with the rolling (roll-adjusted) series $B_m(t)$ and $G_m(t)$, which provide a consistent and comparable historical record across time, constructed using the same rolling convention. Denoting by
\[
r_B(t_k,T_m):=B_m(t_k)-B_m(t_{k-1}),
\qquad
r_G(t_k,T_m):=G_m(t_k)-G_m(t_{k-1}),
\]
the daily variations of the two rolling series, the tenor-specific correlation is estimated over a historical window of length $q=90$ days through the standard sample correlation coefficient, to account of changes across different market regimes and structural conditions. Indeed,
 restricting the estimation to a recent window, we obtain a correlation coefficient that better reflects current market conditions.
 The empirical correlation coefficient is then
\begin{equation*}
\label{rho-hat}
\hat{\rho}_{T_m}^{G,B}
=
\frac{
	\sum_{k=0}^{q-1}
	\bigl(r_G(t_k,T_m)-\bar r_G(T_m)\bigr)
	\bigl(r_B(t_k,T_m)-\bar r_B(T_m)\bigr)
}{
	\sqrt{
		\sum_{k=0}^{q-1}
		\bigl(r_G(t_k,T_m)-\bar r_G(T_m)\bigr)^2
	}
	\sqrt{
		\sum_{k=0}^{q-1}
		\bigl(r_B(t_k,T_m)-\bar r_B(T_m)\bigr)^2
	}
},
\end{equation*}
with
\[
\bar r_G(T_m)=\frac1q\sum_{k=0}^{q-1} r_G(t_k,T_m),
\qquad
\bar r_B(T_m)=\frac1q\sum_{k=0}^{q-1} r_B(t_k,T_m).
\]
In this way, each tenor $T_m$ is associated with a historical correlation estimate consistent with the joint behavior of Brent and Gasoil along the corresponding part of the futures curve. In simulations, for each tenor $T_m$ we use the historical estimate $\hat{\rho}_{T_m}^{G,B}$ from the previous $q$ observations, and we keep this coefficient fixed.

\section{Gasoil option pricing}
\label{Gasoil option pricing}

Once the Brent dynamics has been calibrated, the correction function $h$ has been specified, and the correlation parameters $\rho_{T_m}^{G,B}$ have been estimated, the Gasoil model is fully determined and can be used to  price  Gasoil options, which are not quoted in the market. Since we do not have closed-form pricing formulas, we compute option prices by Monte Carlo simulation. For a fixed quoted maturity $T_m$, we simulate the joint evolution of the two forward prices via the Euler scheme. The Brent process is generated according to the  local-volatility dynamics given in \eqref{Brent_model_proposed}, while the Gasoil process evolves according to  \eqref{Gasoil_model_proposed}, with volatility function $\sigma^G(\cdot,T_m)$ given by \eqref{sigmaG}. At each time step one generates correlated Gaussian shocks for Brent and Gasoil, consistently with the estimated historical correlation coefficient $\rho_{T_m}^{G,B}$, and updates the two processes using the current values of their volatilities.

This leads to the following Euler scheme: 
we consider a discrete Monte Carlo time grid
\[
0=s_0<s_1<\dots<s_n=T_m.
\]
For each time step
\[
\Delta s_k:=s_{k+1}-s_k,
\]
we generate two correlated Gaussian increments
\[
\Delta W_k^B=\sqrt{\Delta s_k}\,Z_k^B,
\qquad
\Delta W_k^G=\sqrt{\Delta s_k}\Bigl(\rho_{T_m}^{G,B}Z_k^B+\sqrt{1-(\rho_{T_m}^{G,B})^2}\,Z_k^G\Bigr),
\]
where $Z_k^B$ and $Z_k^G$ are independent standard normal random variables. The Brent process is then updated according to
\[
B(s_{k+1},T_m)
=
B(s_k,T_m)
+
\sigma^B\bigl(s_k,B(s_k,T_m)\bigr)\Delta W_k^B,
\]
while the Gasoil process is updated as
\begin{align*}
G(s_{k+1},T_m)
=
G(s_k,T_m)
+
\Bigl[
\sigma^B\bigl(s_k,B(s_k,T_m)\bigr)
+
h\bigl(s_k,G(s_k,T_m)-B(s_k,T_m)\bigr)
\Bigr]\Delta W_k^G.
\end{align*}

The price of a European option on Gasoil is then obtained as the risk-neutral expectation of the terminal payoff, approximated by the sample average of a large number $N_{{\mathrm{MC}}}$ of simulated payoffs. {Assuming, for simplicity, a zero interest rate,} the time-zero price of a European option with payoff $\Phi(G(T_m,T_m))$ and maturity $T_m$ is given by
\begin{align*}
O^G(0;K,T_m)
&=\mathbb{E}\bigl[\Phi(G(T_m,T_m))\bigr] \\
&\approx
\frac{1}{N_{\mathrm{MC}}}\sum_{\ell=1}^{N_{\mathrm{MC}}}\Phi\bigl(G^{(\ell)}(T_m,T_m)\bigr).
\end{align*}
In the case of a European call with strike $K$, one has $\Phi(x)=(x-K)^+$, in the case of a European put, $\Phi(x)=(K-x)^+$.  More generally, in the presence of a non-zero risk-free rate $r$, the pricing formula becomes
\begin{align*}
O^G(0;K,T_m)
&=
e^{-rT_m}\,\mathbb{E}\bigl[\Phi(e^{rT_m}G(T_m,T_m))\bigr]\\
&\approx
e^{-rT_m}\frac{1}{N_{\mathrm{MC}}}\sum_{\ell=1}^{N_{\mathrm{MC}}}\Phi\bigl(e^{rT_m}G^{(\ell)}(T_m,T_m)\bigr).
\end{align*}

Once Monte Carlo prices for Gasoil options have been obtained for a range of strikes and maturities $T_1, \dots, T_M$, we invert the Bachelier pricing formula to obtain the model-implied Gasoil volatility surface, and compare it with the one generated by the calibrated Brent model.

We rewrite both surfaces  in terms of moneyness, defined by
\[
k=\frac{K}{F(0,T_m)},
\]
where $F(0,T_m)$ denotes the corresponding initial futures price at maturity $T_m$. For each maturity $T_m$ the smiles are interpolated onto a common grid in the variable $k$ using piecewise cubic Hermite interpolation \cite{fritsch1980monotone}. This yields functions defined on a common moneyness grid
\[
k\mapsto \sigma^{G,\mathrm{mod}}_{{\mathrm{IV}}}(k,T_m),
\qquad
k\mapsto \sigma^{B,\mathrm{mod}}_{{\mathrm{IV}}}(k,T_m),
\]
which can be compared pointwise. We then define the model volatility spread by
\[
\Delta \sigma^{\mathrm{mod}}_{{\mathrm{IV}}}(k,T_m)
:=
\sigma^{G,\mathrm{mod}}_{{\mathrm{IV}}}(k,T_m)-\sigma^{B,\mathrm{mod}}_{{\mathrm{IV}}}(k,T_m).
\]

Finally, let $\sigma^{B,\mathrm{mkt}}_{{\mathrm{IV}}}(k,T_m)$ denote the market-implied Brent volatility surface, expressed in the same moneyness coordinates. We define the final Gasoil volatility surface as
\[
\sigma^{G,\mathrm{final}}_{{\mathrm{IV}}}(k,T_m)
:=
\sigma^{B,\mathrm{mkt}}_{{\mathrm{IV}}}(k,T_m)+\Delta \sigma^{\mathrm{mod}}_{{\mathrm{IV}}}(k,T_m).
\]
In this way, the the Gasoil volatility surface is anchored to the liquid Brent market quotes, while the relative deformation from Brent to Gasoil is provided by the model.

\section{Numerical Results}
\label{sec:numerical_results}

In this section, we  apply the proposed approach to reconstruct numerically the Gasoil implied volatility surface under different market configurations. For both Brent and Gasoil options, market implied volatilities are obtained from ICE (Intercontinental Exchange) end-of-day settlement data at the same valuation date $s_0$. 

We consider a grid of $M=11$ tenors, corresponding to futures maturities $T_1,\dots,T_M$, and a discrete set of strikes for each tenor. The analysis is performed on three representative dates: February 2,
April 21, and May 26, 2026. February 2 is used as a pre-stress benchmark,
corresponding to relatively calm market conditions before the escalation of the
crisis caused by the USA-IRAN conflict. April 21 represents the highly stressed market configuration observed
around late April, while May 26 captures a later phase, when
market conditions had partially stabilized but uncertainty was still significant.

For each date, we consider the ICE Brent and Gasoil market-implied volatility surfaces, denoted by
$$
\sigma^{B,\mathrm{mkt}}_{{\mathrm{IV}}}(k,T),
\qquad
\sigma^{G,\mathrm{mkt}}_{{\mathrm{IV}}}(k,T).
$$

We then compare the following three Gasoil model specifications.

\begin{enumerate}
    \item \textbf{Full model.}  
    The Gasoil volatility is obtained from the Brent volatility plus the crack-spread correction:
    $$
    \sigma^{G}(t,T_m)
    =
    \sigma^B(t,B(t,T_m))
    +
    h\bigl(t,C_m(t)\bigr),
    $$
    where
    $
    C_m(t)=G(t,T_m)-B(t,T_m)
    $
    is the crack spread for tenor $T_m$,
    $\sigma^B$ is the Brent model volatility obtained by the calibration procedure in Section \ref{Calibration} and $h$ is constructed as in Section \ref{h_function_section}.
The corresponding implied volatility surface is denoted by
\begin{equation}\label{sigma1}
\sigma^{G}_{{\mathrm{IV}}}(k,T).
\end{equation}    

    \item \textbf{Brent-only benchmark ($h=0$).}  
    This benchmark removes any Gasoil-specific correction and directly transfers the Brent implied volatility to Gasoil:
    $$
    \sigma^{G,h=0}(t,T_m)
    =
    \sigma^B(t,B(t,T_m)).
    $$
    The corresponding implied volatility surface is denoted by
\begin{equation}\label{sigma2}
    \sigma^{G,h=0}_{{\mathrm{IV}}}(k,T).
\end{equation}    

    \item \textbf{Fixed-spread benchmark.}  
    This benchmark uses a fixed crack-spread correction obtained from the clustering centroids (see \eqref{centroid}). For each tenor $T_m$, the current crack spread $C_m(s_0)$ is computed at the valuation date and used to linearly interpolate the centroid points
    $$
    \bigl\{(c_j^V(T_m),c_j^S(T_m))\bigr\}_{j=1}^{6}.
    $$
    Denoting by $\mathcal I_m$ the corresponding linear interpolant, the fixed correction is defined as
    $$
    h_m^{\mathrm{fix}} = \mathcal I_m\bigl(C_m(s_0)\bigr).
    $$
    The fixed-spread benchmark is therefore
 \begin{equation}\label{sigma3}
    \sigma^{G,\mathrm{fix}}(t,T_m) = \sigma^B(t,B(t,T_m))  +  h^{\mathrm{fix}}_m.\end{equation}    
    This correction is set at the valuation date and remains constant:
unlike the full model, it does not vary with the crack-spread state during the pricing simulation.
    
    The corresponding implied volatility surface is denoted by
    $$
    \sigma^{G,\mathrm{fix}}_{{\mathrm{IV}}}(k,T).
    $$
\end{enumerate}

The numerical analysis proceeds in two steps. First, we compare the Brent and Gasoil market surfaces to highlight the cross-market smile deformation that motivates the model. Second, we compare the observed Gasoil market surface with the three model-implied Gasoil surfaces provided by the full model ($\sigma^{G,\mathrm{mod}}_{{\mathrm{IV}}}(k,T)$), the Brent-only benchmark ($\sigma^{G,h=0}_{{\mathrm{IV}}}(k,T)$)  and the fixed-spread benchmark ($\sigma^{G,\mathrm{fix}}_{{\mathrm{IV}}}(k,T)$).

\subsection{Cross-market volatility structure}

We first compare the Brent market {impled volatility} 
with the Gasoil market {one} 
and with the model-implied Gasoil volatility 
 \eqref{sigma1}. The objective is to show the smile deformation from Brent to Gasoil.

Figures \ref{fig:cross_market_smiles_jan}--\ref{fig:cross_market_smiles_apr} report the comparison for the three selected dates. The smiles are represented on a common grid expressed in standard deviations from the ATM level. More precisely, for each commodity $X\in\{B,G\}$ and tenor $T_m$, we define
$$
z
=
\frac{K-X(0,T_m)}
{\sigma^{X,\mathrm{mkt}}_{\mathrm{IV}}(K_{\mathrm{ATM}},T_m)\sqrt{\tau_m}},
$$
where $X(0,T_m)$ is the futures price, $\sigma^{X,\mathrm{mkt}}_{\mathrm{IV}}(K_{\mathrm{ATM}},T_m)$ is the corresponding ATM implied volatility at time $t=0$, and $\tau_m$ is the time to maturity. Equivalently, a point $z$ on the horizontal axis corresponds to the strike
$$
K
=
X(0,T_m)
+
z\,\sigma^{X,\mathrm{mkt}}_{\mathrm{IV}}(K_{\mathrm{ATM}},T_m)\sqrt{\tau_m}.
$$
This normalization allows us to compare the relative shape of Brent and Gasoil smiles on the same scale.

\begin{figure}[H]
    \centering
    \includegraphics[width=0.95\textwidth]{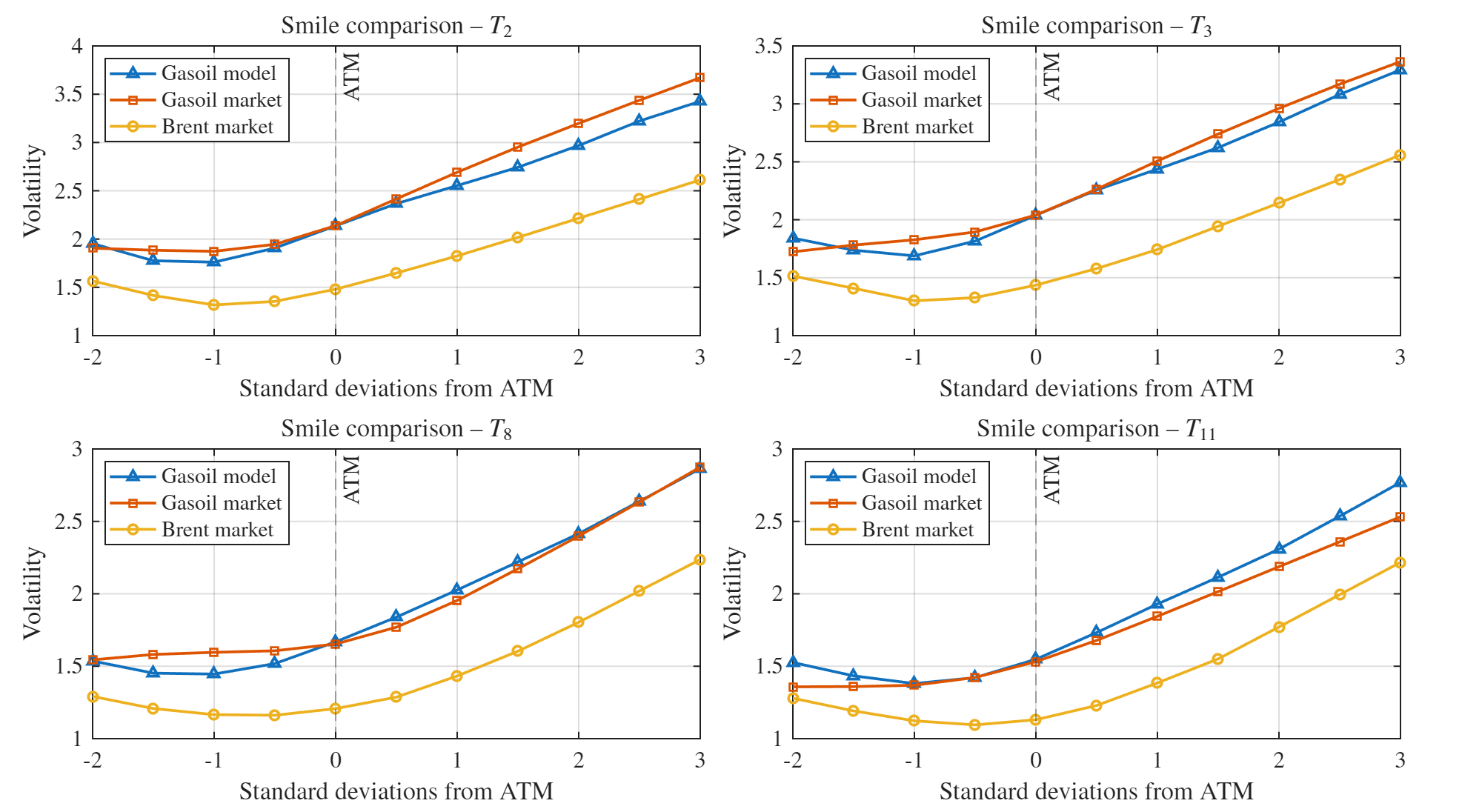}
    \caption{Brent market smiles, Gasoil market smiles and model-implied Gasoil smiles on February 2, 2026. }
    \label{fig:cross_market_smiles_jan}
\end{figure}

\begin{figure}[H]
    \centering
    \includegraphics[width=0.95\textwidth]{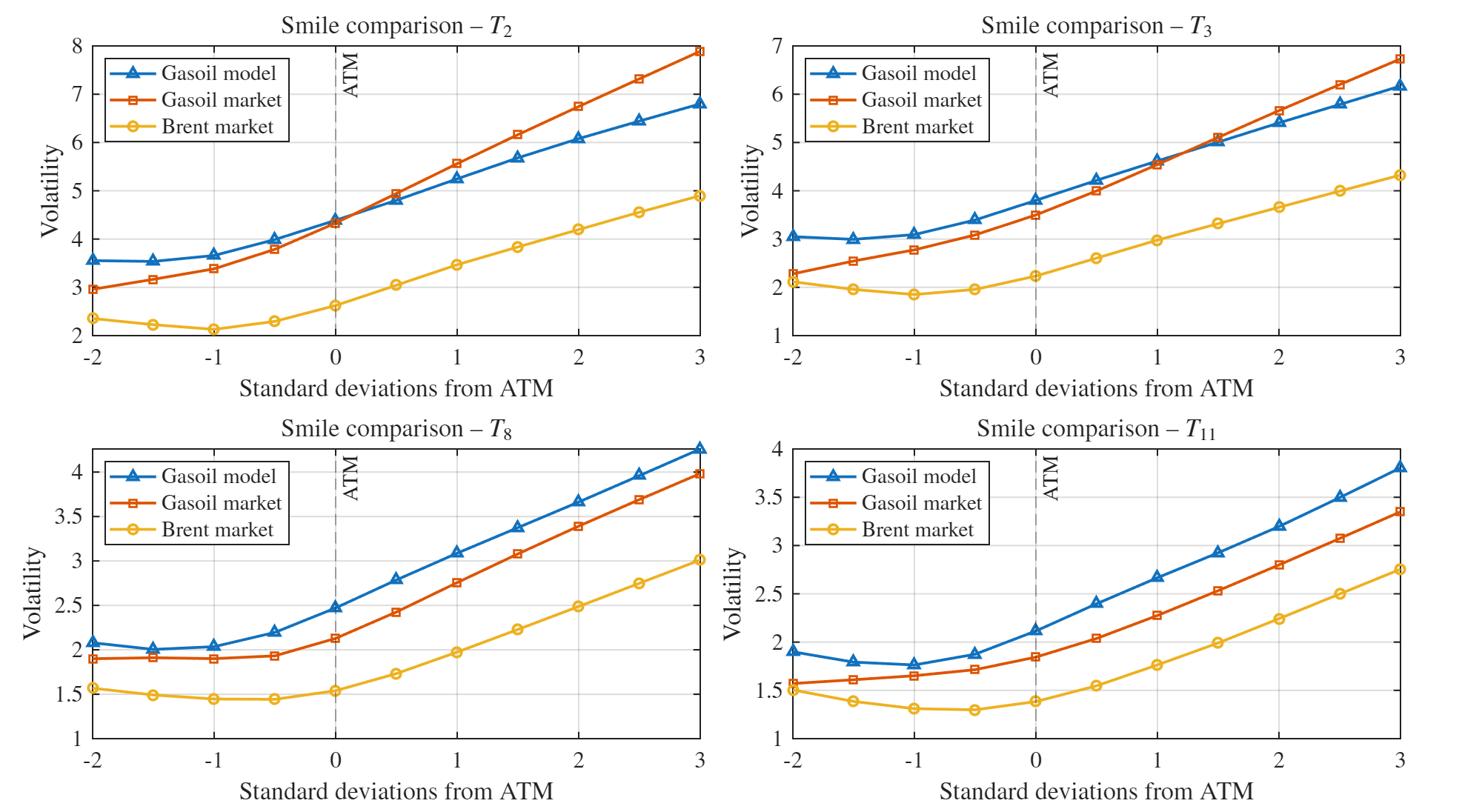}
    \caption{Brent market smiles, Gasoil market smiles and model-implied Gasoil smiles on April 21, 2026.}
    \label{fig:cross_market_smiles_feb}
\end{figure}

\begin{figure}[H]
    \centering
    \includegraphics[width=0.95\textwidth]{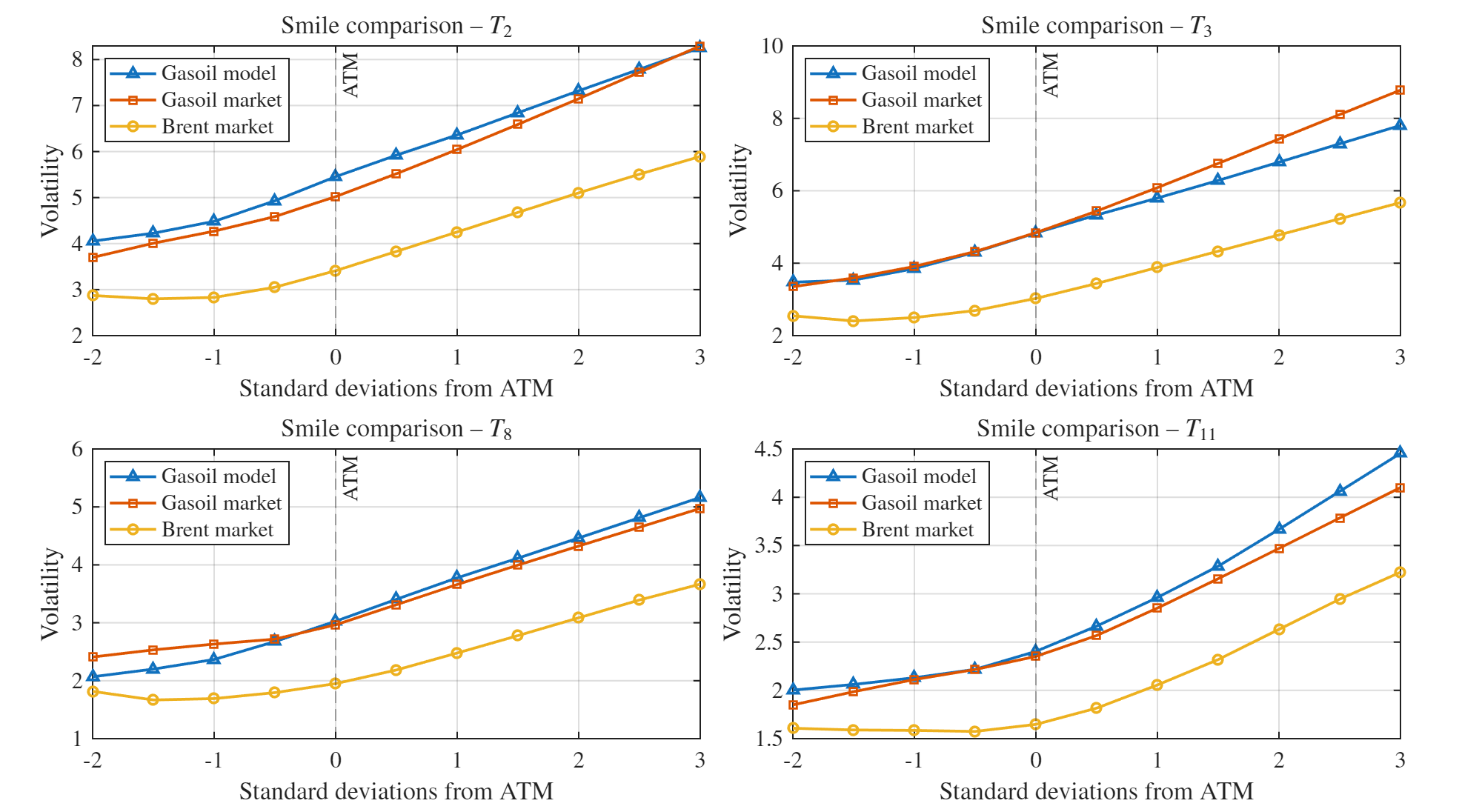}
    \caption{Brent market smiles, Gasoil market smiles and model-implied Gasoil smiles on May 26, 2026. }
    \label{fig:cross_market_smiles_apr}
\end{figure}

The figures show that the Gasoil market volatility is not a direct translation of the Brent market one. Gasoil smiles display indeed a more pronounced right wing, meaning that a direct transfer of the Brent volatility surface would fail to capture the actual shape of the Gasoil smile.

The model-implied Gasoil smiles capture this effect reasonably well. They generate a higher volatility level and a stronger right wing than Brent, in line with the main empirical feature observed in the Gasoil market. This supports the introduction of correction in the Gasoil dynamics, and the one based on the a crack-spread seems to work well.

\subsection{Comparison with the Gasoil market surface}

We now compare the observed Gasoil market surface with the three Gasoil model specifications already introduced: the full model \eqref{sigma1}, the Brent-only benchmark \eqref{sigma2}, and the fixed-spread benchmark \eqref{sigma3}.

Figures \ref{fig:gasoil_model_comparison_jan}--\ref{fig:gasoil_model_comparison_apr} report the comparison for the three selected dates and selected tenors.

\begin{figure}[H]
    \centering
    \includegraphics[width=0.95\textwidth]{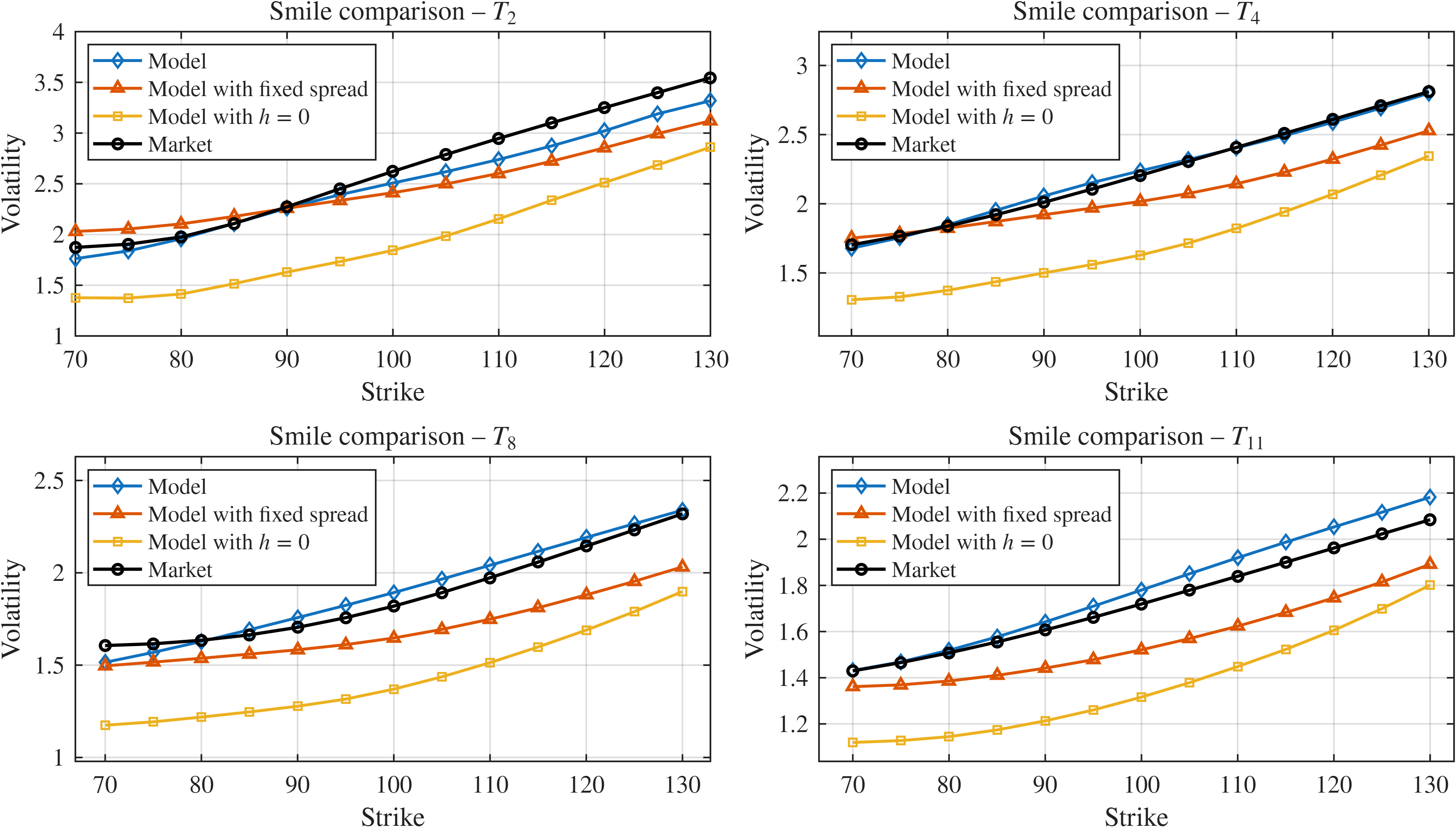}
    \caption{Comparison between the Gasoil market smile and the three model specifications on February 2, 2026.}
    \label{fig:gasoil_model_comparison_jan}
\end{figure}

\begin{figure}[H]
    \centering
    \includegraphics[width=0.95\textwidth]{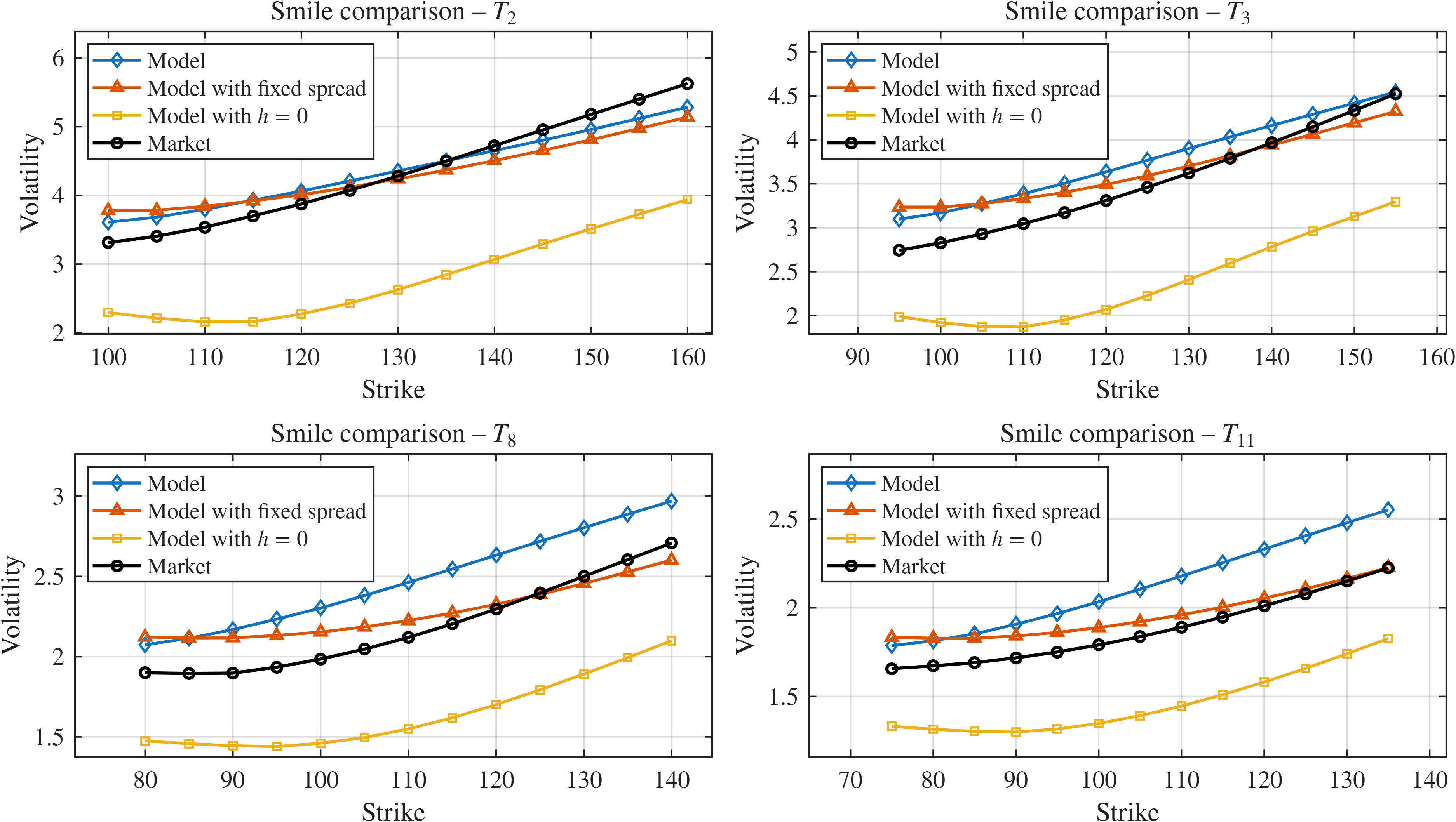}
    \caption{Comparison between the Gasoil market smile and the three model specifications on April 21, 2026.}
    \label{fig:gasoil_model_comparison_feb}
\end{figure}

\begin{figure}[H]
    \centering
    \includegraphics[width=0.95\textwidth]{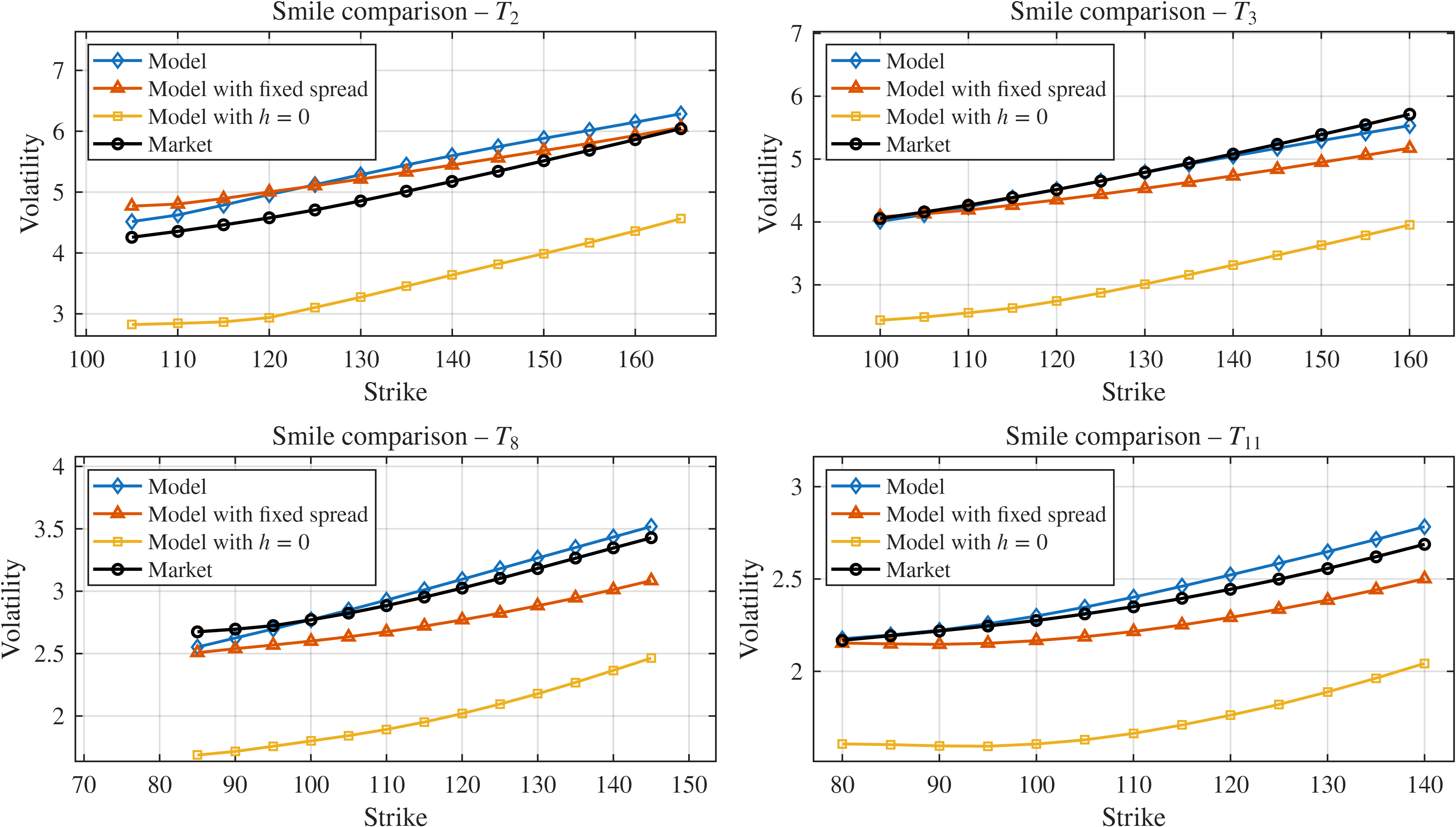}
    \caption{Comparison between the Gasoil market smile and the three model specifications on May 26, 2026.}
    \label{fig:gasoil_model_comparison_apr}
\end{figure}

The comparison shows that the Brent-only benchmark $h=0$ is too restrictive: it stays close to the Brent volatility structure and does not reproduce the Gasoil smile accurately. The fixed-spread benchmark improves on this direct transfer by adding a correction linked to the current crack-spread level, but this correction is static and therefore cannot always reproduce the full curvature of the smile.

The full model is more flexible because it adjusts both the level and the shape of the Gasoil smile, including the stronger right-wing behaviour observed in the market. 
\begin{table}[H]
\centering
\begin{tabular}{l l c c c}
\toprule
Date & Specification & RMSE & MAE & Relative error \\
\midrule
February 2, 2026
& Full model         & 0.1132 & 0.0760 & 3.28\% \\
& Brent-only $h=0$   & 0.5215 & 0.5002 & 23.76\% \\
& Fixed spread       & 0.2310 & 0.1983 & 8.98\% \\
\midrule
April 21, 2026
& Full model         & 0.2959 & 0.2735 & 10.97\% \\
& Brent-only $h=0$   & 0.9674 & 0.8479 & 27.79\% \\
& Fixed spread       & 0.2293 & 0.1963 & 7.91\% \\
\midrule
May 26, 2026
& Full model         & 0.2174 & 0.1630 & 3.87\% \\
& Brent-only $h=0$   & 1.3472 & 1.2872 & 33.67\% \\
& Fixed spread       & 0.3696 & 0.3260 & 8.33\% \\
\bottomrule
\end{tabular}
\caption{Global error metrics against the Gasoil market surface for the three model specifications.}
\label{tab:gasoil_model_errors}
\end{table}

Table \ref{tab:gasoil_model_errors} summarizes the global errors  for the three specifications and Figures \ref{fig:relative_error_jan}--\ref{fig:relative_error_apr} show the relative error by tenor, in order to show where each specification performs better or worse along the term structure.

\begin{figure}[H]
    \centering
    \includegraphics[width=0.8\textwidth]{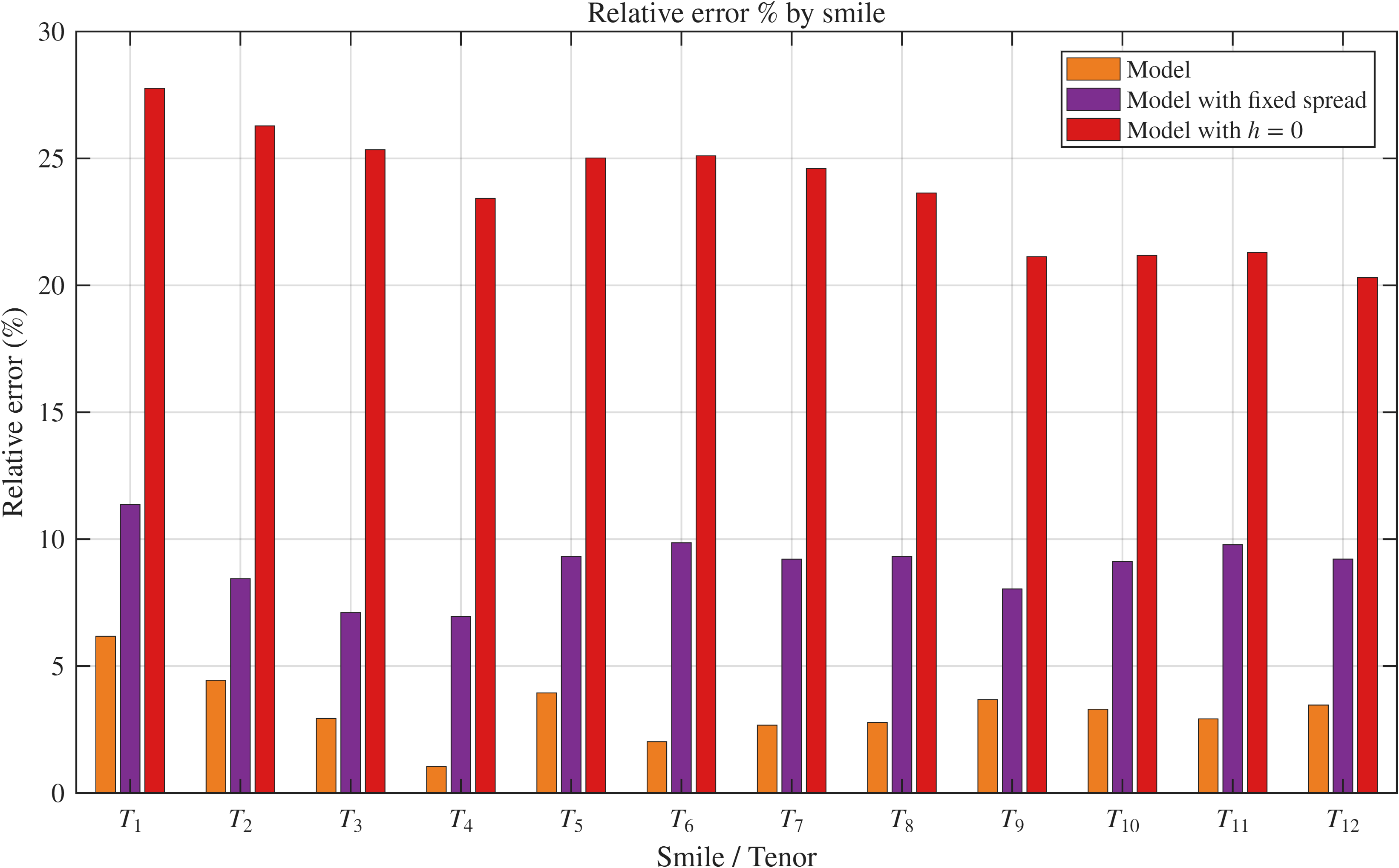}
    \caption{Relative error by tenor for the three Gasoil model specifications on  February 2, 2026}   \label{fig:relative_error_jan}
\end{figure}

\begin{figure}[H]
    \centering
    \includegraphics[width=0.8\textwidth]{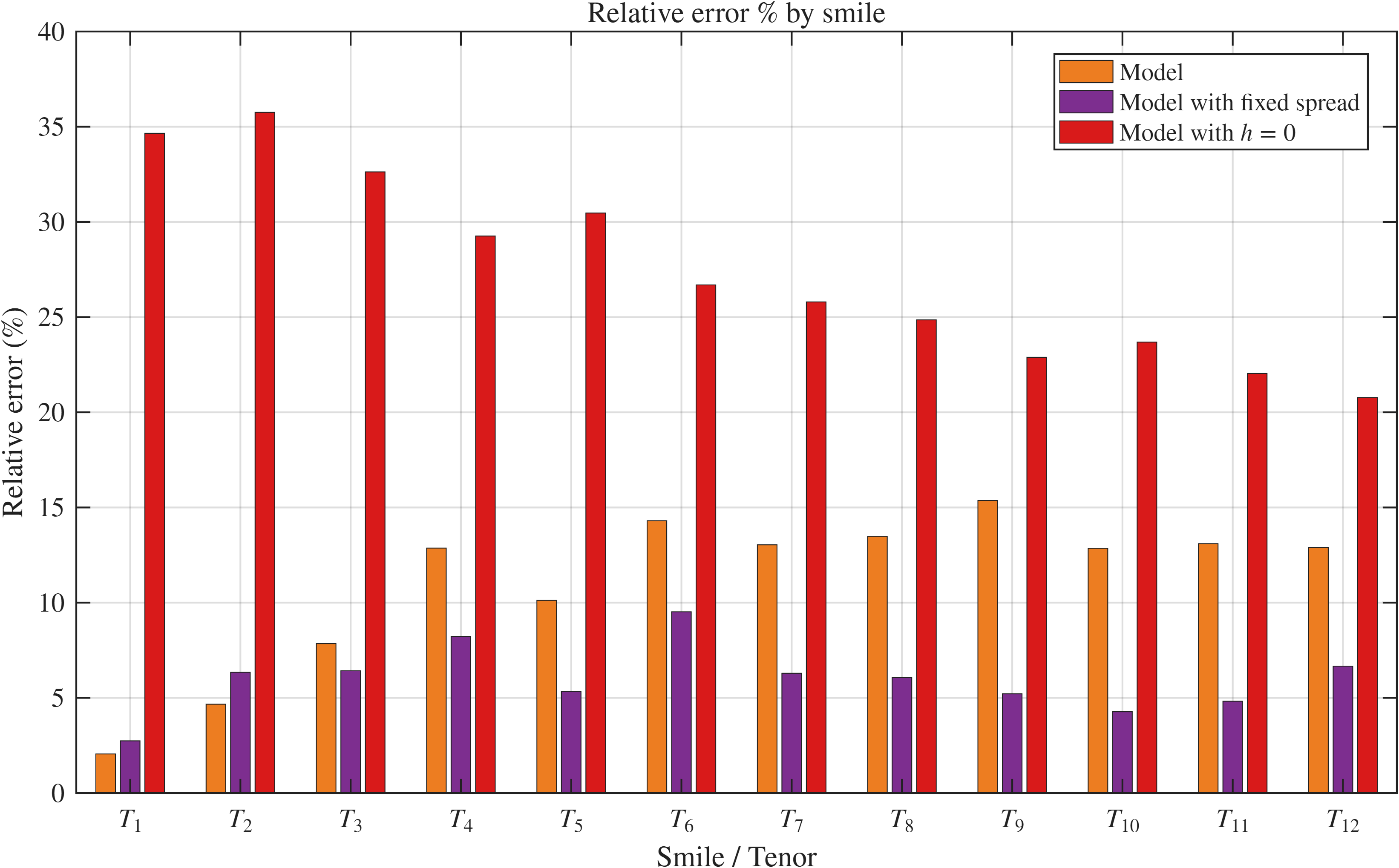}
    \caption{Relative error by tenor for the three Gasoil model specifications on April 21, 2026.}
    \label{fig:relative_error_feb}
\end{figure}

\begin{figure}[H]
    \centering
    \includegraphics[width=0.8\textwidth]{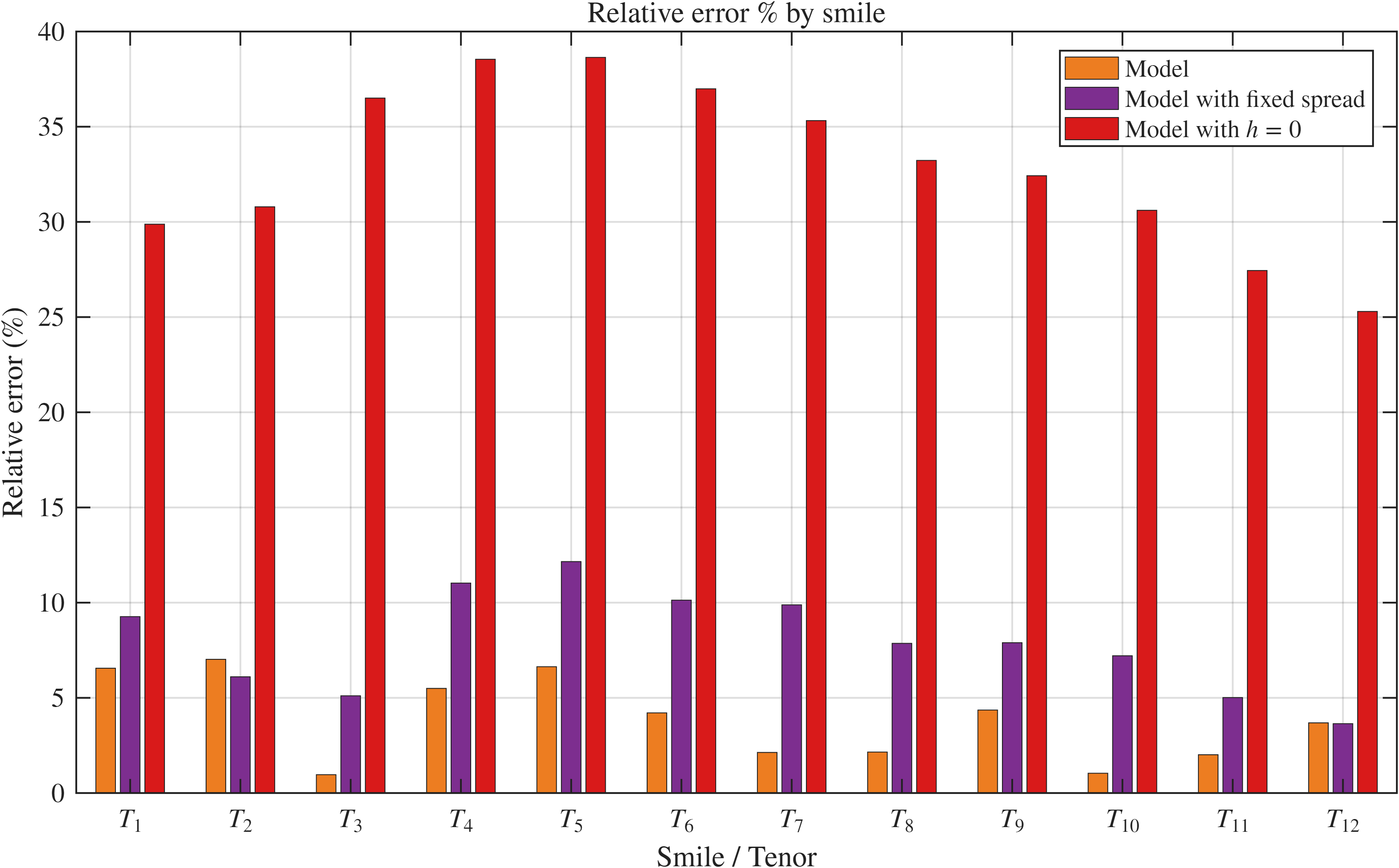}
    \caption{Relative error by tenor for the three Gasoil model specifications on May 26, 2026.}
    \label{fig:relative_error_apr}
\end{figure}

The global and tenor-level errors give the same conclusion. The Brent-only benchmark $h=0$ produces the largest errors over almost the whole term structure, confirming that Brent volatility cannot be transferred directly to Gasoil.

Both corrected specifications improve the fit substantially. On February 2 and May 26, the full model achieves the lowest global errors and generally performs best across tenors. On April 21, the fixed-spread benchmark is slightly more accurate in global terms and outperforms the full model for several maturities. This suggests that, on that date, the interpolated crack-spread correction already provides a very effective level adjustment.

Overall, the main improvement comes from adding a Gasoil-specific correction. The full model provides the most flexible mechanism to adjust both the level and the shape of the Gasoil smile, while the fixed-spread benchmark confirms that the crack-spread component contains relevant information for the Gasoil volatility surface.

Finally, we remark that the Gasoil option surface is not used as an input in the model calibration. It is used only ex post, as a market benchmark to evaluate the reconstructed surface. The model-implied Gasoil surface is obtained from liquid Brent option quotes, which determine the Brent volatility structure, and from Gasoil futures prices, which enter through the crack-spread correction. When considering this fact, we can say that the fit between model and marked implied Gasoil implied volatilities is remarkably good.

\printbibliography

\end{document}